\documentclass[a4paper]{article}

\usepackage[english]{babel}
\usepackage[utf8]{inputenc}
\usepackage[T1]{fontenc}
\usepackage{amsmath}
\usepackage{graphicx}
\usepackage[colorinlistoftodos]{todonotes}
\usepackage[a4paper,left=2.5cm,right=2.5cm,top=2.cm,bottom=2cm]{geometry}
\usepackage{authblk}
\usepackage{dsfont}
\usepackage{txfonts}
\usepackage{bbold} 
\usepackage{amsopn}  
\usepackage{amssymb}
\usepackage{graphicx}
\usepackage{array}
\usepackage{soul}
\usepackage{xparse}
\newcolumntype{M}[1]{>{\centering\arraybackslash}m{#1}}
\newtheorem{example}{Example}
\newtheorem{proposition}{Proposition}
\newtheorem{definition}{Definition}
\newtheorem{lemma}{Lemma}
\newtheorem{corollary}{Corollary}
\newenvironment{proof}{\noindent\textit{Proof~~}}
{\nolinebreak[4]\hfill$\blacksquare$\\\par}
\title{Permissible four-strategy quantum extensions of classical games}

\author{
  Piotr Frąckiewicz\\
  \textit{Institute of Exact and Technical Sciences, Pomeranian University in Słupsk, Poland} \\
  piotr.frackiewicz@upsl.edu.pl
  \and
  Anna Gorczyca-Goraj\\
  \textit{Department of Operations Research,
        University of Economics in Katowice, Poland} \\
        anna.gorczyca-goraj@uekat.pl
    \and
  Marek Szopa\\
  \textit{Department of Operations Research,
        University of Economics in Katowice, Poland} \\
        marek.szopa@uekat.pl
}

\date{\today}

\begin{document}
\maketitle

\begin{abstract}
The study focuses on strategic-form games extended in the Eisert-Wilkens-Lewenstein scheme by two unitary operations. Conditions are determined under which the pair of unitary operators, along with classical strategies, form a game invariant under isomorphic transformations of the input classical game. These conditions are then applied to determine these operators, resulting in five main classes of games satisfying the isomorphism criterion, and a theorem is proved providing a practical criterion for this isomorphism. The interdependencies between different classes of extensions are identified, including limit cases in which one class transforms into another.
\end{abstract}
\section{Introduction}
The burgeoning field of quantum game theory, particularly through the Eisert-Wilkens-Lewenstein (EWL) approach  \cite{eisert_quantum_1999} presents a novel paradigm for understanding strategic interactions in quantum information processing and decision-making. This innovative perspective not only extends classical game theory into the quantum domain but also uncovers new dimensions of strategic complexity and potential advantages inherent in quantum mechanics. The EWL approach, foundational to quantum game theory, has facilitated the translation of classical game models into the quantum framework, enabling the analysis of games with superposed states and entanglement. By incorporating quantum strategies, which are essentially operations on quantum states, this approach allows for the exploration of extensions that are unattainable within the classical strategic landscape. Particularly, the study of mixed strategies involving several pure quantum strategies opens up a plethora of strategic options and outcomes, potentially surpassing the limitations of classical mixed strategies.

Despite more than 20 years of research on the EWL scheme, precise conditions indicating appropriate unitary strategies in the EWL scheme have not been defined. A natural choice is the special unitary group $\mathsf{SU}(2)$ \cite{benjamin_comment_2001}. These operations are closed with respect to multiplication. In addition, the extension of $\mathsf{SU}(2)$ to the whole unitary group $\mathsf{U}(2)$ adds only strategies that are payoff equivalent to operations of $\mathsf{SU}(2)$. The set $\mathsf{SU}(2)$ has another important property. Namely, having two classical $2\times 2$ games differing in the order of the rows or columns of the bimatrix, we can be sure that an EWL scheme with a set of $\mathsf{SU}(2)$ strategies will always generate identical games \cite{frackiewicz_strong_2016}. More precisely, the EWL scheme with the strategy set $\mathsf{SU}(2)$ implies equivalent games whenever the considered classical games are isomorphic. The independence of the EWL game with respect to isomorphic transformations of the classical game guarantees that any game theory problem expressed in EWL terms will generate exactly the same game. 

An example that perfectly illustrates this point is the Prisoner's Dilemma game -- a problem considered in the pioneering work \cite{eisert_quantum_1999}. There, a bimatrix representing the Prisoner's Dilemma problem was studied using a certain subset of two-parameter unitary operations. The authors showed that in the set of two-parameter unitary operations, there exists a strategy profile that is a Pareto-optimal Nash equilibrium. While it may seem that the approach presented in \cite{eisert_quantum_1999} resolves the Prisoner's Dilemma in the quantum domain, doubts arise due to the use of a proper subset of $\mathsf{SU}(2)$ instead of the entire set. As we demonstrated in \cite{frackiewicz_quantum_2016,frackiewicz_permissible_2024}, interchanging rows or columns in the bimatrix significantly impacts the set of Nash equilibria in the corresponding EWL game when some specific unitary operations are chosen. In other words, considering the same problem in game theory, up to the order in which we write rows and columns, we obtain completely different solutions in quantum games. This fact clearly indicates that the requirement for the invariance of the scheme under isomorphic transformations of a classical game is necessary in order to consider the EWL scheme as a valid extension of the classical game.

The extensions of classical games considered in this paper are based on equipping players of the classical $2\times 2$ game with a set of four unitary strategies. Among these, two correspond exactly to the initial classical strategies, while the other two are appropriate extensions. On these additional quantum strategies, however, the players act in a classical way and mix them with other strategies in any way they wish. This approach allows them to develop mixed strategies analogous to those in the classical game, but now enriched by the extended game. It could be argued that classical players might not recognize the quantum genesis of these supplementary strategies, but still use them effectively. Extending the game in this manner allows players to achieve outcomes that are superior to those achievable within the confines of the classical game alone. Here, 'superior' refers to scenarios where, for instance, players seeking a Nash equilibrium can achieve outcomes that are closer to Pareto optimality than what is possible in the classical game.

The aim of the current work is to determine all such extensions that meet the condition of invariance towards isomorphic transformations of the classical game, i.e. that are a faithful reflection of the classical game. This implies that players are equipped with quantum strategies that exactly mimic the way classical strategies work, ensuring that the added strategies do not change the fundamental structure of the game. Consequently, different isomorphic versions of the same classical game should yield identical, up to isomorphism, versions of the extended game. The significance of this approach in quantum game theory lies in its recognition of a condition that has been historically disregarded. In fact, deviation from this principle was a common trend in the evolution of quantum game theory to date \cite{eisert_quantum_1999,khan_quantum_2018,naskar_quantum_2021}. However, this results in unjustified ambiguity of extensions, which varies depending on the form of the initial game.

In the first part of section 2, we present necessary concepts of classical game theory, game isomorphism and payoff equivalent strategies accompanied by illustrative examples. This section's second part offers a concise overview of the EWL quantization method, highlighting and using examples to illustrate the concept of payoff equivalent quantum strategies. This chapter will also present two essential theorems. The first one demonstrates the form of transformations of unitary strategies that ensure the invariance the quantum game payoffs for specific isomorphic forms of the classical game. The second, crucial theorem formulates a criterion - a necessary condition for the isomorphicity of the EWL extensions, generated by a finite set of strategies, of isomorphic versions of the same classical game. In the third section, we provide reader with the acquired solutions for this criterion, which include five classes of possible parameters for strategies that result in acceptable extensions. The fourth section introduces a practical way to determine the invariance of bimatrices for extensions, using four isomorphic forms of the classical game. Furthermore, we present the explicit form of the five extensions of the classical game using $4\times 4$ bimatrices that are invariant to isomorphic transformations of the classical game and interdependencies between different classes of extensions, i.e. limiting cases in which one class transforms into another.

\vspace{12pt}

\section{Preliminaries}
Our article is self-contained as we introduce essential concepts from both game theory and quantum game theory in this section. 
\subsection{Classical game theory}
We focus on one of the primary types of games in non-cooperative game theory, namely, in strategic form games \cite{maschler_game_2020}.
\begin{definition}
A game in strategic form is a triple $\Gamma = (N, (S_{i})_{i\in N}, (u_{i})_{i\in N})$ in which 
\begin{enumerate}
\item[(i)] $N = \{1,2, \dots, p\}$ is a finite set of players;
\item[(ii)] $S_{i}$ is the set of strategies of player $i$, for each player $i\in N$;
\item[(iii)] $u_{i}\colon S_{1}\times S_{2} \times \cdots \times S_{p} \to \mathds{R}$ is a function associating each vector of strategies $s = (s_{i})_{i\in N}$ with the payoff $u_{i}(s)$ to player $i$, for every player $i\in N$.
\end{enumerate}
\end{definition}
In the case of a two-player scenario, a strategic-form game can be represented by a bimatrix,
\begin{equation}
\begin{pmatrix}
\Delta_{00} & \Delta_{01} \\
\Delta_{10} & \Delta_{11}
\end{pmatrix}, ~~\text{where}~~\Delta_{ij} = (a_{ij}, b_{ij}) ~\text{and}~ a_{ij}, b_{ij} \in \mathds{R}.
\end{equation}
The rows and columns of the bimatrix are then identified with the strategies of the first and second player, respectively. Each entry in the bimatrix is a pair of payoffs for the players.

Now, we recall the notion of isomorphism as it applies to strategic-form games. The definition is based on \cite{gabarro_complexity_2007}, (see also \cite{nash_non-cooperative_1951,peleg_canonical_1999,sudholter_canonical_2000}). Since we restrict ourselves to two-person games, we provide a simplified version that does not include numbering of the players. 
\begin{definition}
Given $\Gamma = (N, (S_{i})_{i\in N}, (u_{i})_{i\in N})$ and $\Gamma' = (N, (S'_{i})_{i\in N}, (u'_{i})_{i\in N})$, let $(\varphi_{i})_{i\in N}$ be a collection of bijections $\varphi_{i}$ from $S_{i}$ to $S'_{i}$. A collection $(\varphi_{i})_{i\in N}$ is a strong isomorphism between $\Gamma$ and $\Gamma'$ if relation 
\begin{equation}\label{isocondition}
u_{i}\bigl((s_{1}, s_{2}, \dots, s_{p})\bigr) = u'_{i}\bigl((\varphi_{1}(s_{1}), \varphi_{2}(s_{2}), \dots, \varphi_{p}(s_{p}))\bigr)
\end{equation}
holds for each $i\in N$ and each strategy profile $(s_{1}, s_{2}, \dots, s_{p}) \in S_{1} \times S_{2} \times \cdots \times S_{p}$. In this case, the games $\Gamma$ and $\Gamma'$ are referred to as strongly isomorphic.
\end{definition}
\begin{example}
Let us consider a $2\times 2$ bimatrix game
\begin{equation}
\label{bimatrixiso}
\Gamma = \bordermatrix{ & C & D\cr
A & \Delta_{00} & \Delta_{01} \cr
B & \Delta_{10} & \Delta_{11}
}.
\end{equation}
There are three different bimatrix games 
\begin{equation}
\label{bimatrix2}
\bordermatrix{ & C' & D'\cr
A' & \Delta'_{00} & \Delta'_{01} \cr
B' & \Delta'_{10} & \Delta'_{11}
},
\end{equation}
that are isomorphic to (\ref{bimatrixiso}):
\begin{align}\label{gamma1}
&\text{game with swapped rows} \quad  \Gamma^1 = \bordermatrix{ & C & D\cr
B & \Delta_{10} & \Delta_{11} \cr
A & \Delta_{00} & \Delta_{01}
},\\ \label{gamma2}
&\text{game with swapped columns} \quad  \Gamma^2 =\bordermatrix{ & D & C\cr
A & \Delta_{01} & \Delta_{00} \cr
B & \Delta_{11} & \Delta_{10}
}, \\ \label{gamma3}
&\text{game with swapped rows and columns} \quad  \Gamma^{3} = \bordermatrix{ & D & C\cr
B & \Delta_{11} & \Delta_{10} \cr
A & \Delta_{01} & \Delta_{00}
}.
\end{align}
For instance, the game (\ref{gamma1}) is isomorphic to (\ref{bimatrixiso}) as a result of applying two bijections
\begin{equation}
\varphi_{1} = (A \to B', B \to A'), \quad \varphi_{2} = (C \to C', D \to D')
\end{equation}
for (\ref{bimatrixiso}) and (\ref{bimatrix2}). Then the payoff functions are
    \begin{equation}
    \begin{aligned}
&u_{i}(A,C) = u'_{i}(B',C') \Leftrightarrow \Delta_{00} = \Delta'_{10}, \\
&u_{i}(A,D) = u'_{i}(B',D') \Leftrightarrow \Delta_{01} = \Delta'_{11}, \\
&u_{i}(B,C) = u'_{i}(A',C') \Leftrightarrow \Delta_{10} = \Delta'_{00}, \\
&u_{i}(B,D) = u'_{i}(A',D') \Leftrightarrow \Delta_{11} = \Delta'_{01}.
\end{aligned}
    \end{equation}
\end{example}
We denote the Cartesian product of all the strategy sets $S_{j}$ except for the set $S_{i}$ by $S_{-i} = \varprod_{j\ne i}S_{j}$. An element in $S_{-i}$ will be denoted by $s_{-i}$. The next concept we use in our work is the notion of strategy equivalence with respect to payoffs \cite{myerson_game_1991}.
\begin{definition}\label{meyerson}
Given a strategic-form game $\Gamma = (N, (S_{i})_{i\in N}, (u_{i})_{i\in N})$, two pure strategies $s'_{i}$, $s''_{i}$ are payoff equivalent if
\begin{equation}
u_{j}(s'_{i}, s_{-i}) = u_{j}(s''_{i}, s_{-i}) 
\end{equation}
for all $s_{-i} \in S_{-i}$ and for all $j \in N$.
\end{definition}
Let $s_{i} \in S_{i}$. We denote by $[s_{i}]$ the set of pure strategies $s'_{i}$ that are payoff equivalent with $s_{i}$.\\

In bimatrix games, payoff equivalent pure strategies can be recognized as two identical rows or two identical columns. Surprisingly, this concept, originating from classical game theory, appears to have nontrivial significance in quantum games, where unitary operations play the role of strategies. The significance of this concept will be presented during the discussion of the quantum game scheme in the next subsection.

Another concept of game theory that we use is Nash equilibrium \cite{nash_equilibrium_1950}, \cite{maschler_game_2020}. 
\begin{definition}
A strategy vector $s^* = (s^*_{1}, s^*_{2}, \dots, s^*_{n})$ is a Nash equilibrium if for each player $i\in N$ and each strategy $s_{i} \in S_{i}$ the following inequalities are satisfied
\begin{equation*}
u_{i}(s^*) \geq u_{i}(s_{i}, s^*_{-i}). 
\end{equation*}
\end{definition}
To put it another way, a Nash equilibrium is a strategy profile at which no player has a profitable deviation when all the remaining players do not change their strategies.
\subsection{Eisert-Wilkens-Lewenstein quantum game scheme}
In this section, we recall the EWL-type scheme \cite{eisert_quantum_1999} for $2\times 2$ games (\ref{bimatrixiso}).
The special unitary group $\mathsf{SU}(2)$ plays the role of strategy sets in the EWL scheme. The commonly used parametrization of the unitary strategy from $\mathsf{SU}(2)$ is
\begin{equation}\label{Umatrix}
U(\theta, \alpha, \beta) = \begin{pmatrix}
e^{i\alpha}\cos\frac{\theta}{2} & ie^{i\beta}\sin\frac{\theta}{2} \\ ie^{-i\beta}\sin\frac{\theta}{2} & e^{-i\alpha}\cos\frac{\theta}{2}
\end{pmatrix}, \quad \theta\in [0,\pi], \quad \alpha, \beta \in [0,2\pi).
\end{equation}
By choosing their $U_{1}(\theta_{1}, \alpha_{1}, \beta_{1})$ and $U_{2}(\theta_{2}, \alpha_{2}, \beta_{2})$ strategies, players determine the final state of the game
\begin{equation}
|\Psi\rangle = J^{\dag}\left(U_{1}(\theta_{1}, \alpha_{1}, \beta_{1})\otimes U_{2}(\theta_{2}, \alpha_{2}, \beta_{2}) \right)J|00\rangle, 
\end{equation}
where $J = (I\otimes I + i\sigma_{x}\otimes \sigma_{x})/\sqrt{2}$.

The payoff $u_{i}$ for player $i$ is defined as
\begin{equation}\label{payoffM}
u_{i}(U_{1}(\theta_{1}, \alpha_{1}, \beta_{1}), U_{2}(\theta_{2}, \alpha_{2}, \beta_{2})) = \langle \Psi|M_{i}|\Psi\rangle,
\end{equation}
where $M_{i}$ for $i=1,2$ are the measurement operators determined by the bimatrix payoffs $\Delta_{ij} = (a_{ij}, b_{ij})$,
\begin{equation}
    M_{1} = \sum_{i,j\in {0,1}} a_{ij}|ij\rangle \langle ij|, \quad M_{2} = \sum_{i,j\in {0,1}} b_{ij}|ij\rangle \langle ij|.
\end{equation}
Using formula (\ref{payoffM}), we can determine the explicit form of the pair of players' payoffs,
\begin{align}\label{generalEWLpayoff}
    &(u_{1}, u_{2})(U_{1}(\theta_{1}, \alpha_{1}, \beta_{1}), U_{2}(\theta_{2}, \alpha_{2}, \beta_{2})) \nonumber\\ &\quad = \Delta_{00}\left(\cos(\alpha_{1} + \alpha_{2})\cos\frac{\theta_{1}}{2}\cos\frac{\theta_{2}}{2} + \sin(\beta_{1} + \beta_{2})\sin\frac{\theta_{1}}{2}\sin\frac{\theta_{2}}{2}\right)^2 \nonumber\\
    &\quad + \Delta_{01}\left(\cos(\alpha_{1} - \beta_{2})\cos\frac{\theta_{1}}{2}\sin\frac{\theta_{2}}{2} + \sin(\alpha_{2} - \beta_{1})\sin\frac{\theta_{1}}{2}\cos\frac{\theta_{2}}{2}\right)^2 \nonumber \\ 
    &\quad + \Delta_{10}\left(\sin(\alpha_{1} - \beta_{2})\cos\frac{\theta_{1}}{2}\sin\frac{\theta_{2}}{2} + \cos(\alpha_{2} - \beta_{1})\sin\frac{\theta_{1}}{2}\cos\frac{\theta_{2}}{2}\right)^2 \nonumber \\ 
    &\quad + \Delta_{11}\left(\sin(\alpha_{1} + \alpha_{2})\cos\frac{\theta_{1}}{2}\cos\frac{\theta_{2}}{2} - \cos(\beta_{1} + \beta_{2})\sin\frac{\theta_{1}}{2}\sin\frac{\theta_{2}}{2}\right)^2.
\end{align}

\begin{figure}[t]
    \centering 
\includegraphics[width=4.5in]{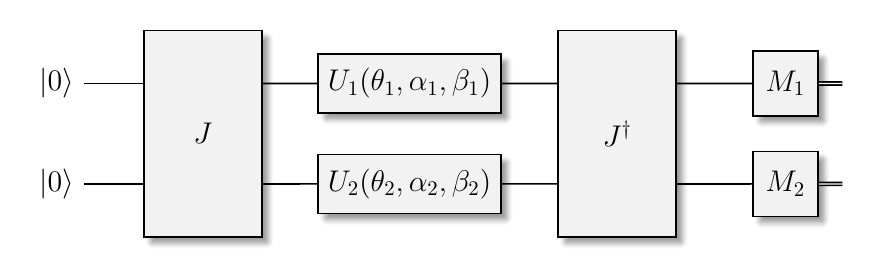}
\caption{The EWL scheme}
\end{figure}

The dependence of payoffs on players' strategies is not a one-to-one relationship - multiple strategies can result in the same payoff combinations. Below are examples of payoff-equivalent strategies in the EWL game.
\begin{example}\label{ex2}
A natural example of payoff equivalent strategies in the EWL scheme is given by unitary operations differing by a global phase factor. i.e., if $U \in \mathsf{SU}(2)$ then $e^{i\delta}U \in [U]$. This property arises from the construction of the payoff functions (\ref{payoffM}) in the EWL scheme, which is based on a quantum measurement.
\end{example}
\begin{example}\label{ex3}
Let us consider unitary operators of the form:
\begin{equation}
U_{1} = U\left(\frac{\pi}{2}, \frac{\pi}{2}, \frac{\pi}{2}\right) = \frac{1}{\sqrt{2}}\begin{pmatrix}
i & -1 \\ 
1 & -i
\end{pmatrix}, \quad U_{2} =U\left(\frac{\pi}{2}, \frac{3\pi}{2}, \frac{\pi}{2}\right) = \frac{1}{\sqrt{2}}\begin{pmatrix}
-i & -1 \\ 
1 & i
\end{pmatrix}.
\end{equation}
Then $U_{1}$ and $U_{2}$ are payoff equivalent for player 1 in the EWL game with strategy sets $S_{1} = S_{2} = \{I, iX, U_{1}, U_{2}\}$ as $u_{1}(U_{1}, s_{2}) = u_{1}(U_{2}, s_{2})$ for any $s_{2} \in S_{2}$. Indeed, 
\begin{equation}
\begin{aligned}
&u_{1}(U_{1}, I) = u_{1}(U_{2}, I) = \frac{a_{01} + a_{11}}{2}, \\
&u_{1}(U_{1}, iX) = u_{2}(U_{2}, iX) = \frac{a_{00} + a_{10}}{2} \\ 
&u_{1}(U_{1}, U_{1}) = u_{2}(U_{2}, U_{1}) = \frac{a_{00}+a_{01} + a_{10} + a_{11}}{4} \\ 
&u_{1}(U_{1}, U_{2}) = u_{2}(U_{2}, U_{2}) = \frac{a_{00}+a_{01} + a_{10} + a_{11}}{4}.
\end{aligned}
\end{equation}

\end{example}
In the EWL scheme, one can therefore identify payoff equivalent profiles. Due to the form of the payoff function (\ref{generalEWLpayoff}) by applying trigonometric reduction formulas, we can easily show that each strategy profile determines a class of strategy pairs generating the same payoff vector. We can formulate this property with the following lemma:
\begin{lemma}\label{lemma1} 
Let $(u_{1}, u_{2})$ be a payoff vector given by (\ref{generalEWLpayoff}). Then
\begin{itemize}
\item adding $\pi$ to two values of $\alpha_{1}, \alpha_{2}, \beta_{1}, \beta_{2}$
\item adding $\pi/2$ to all values of $\alpha_{1}, \alpha_{2}, \beta_{1}, \beta_{2}$
\end{itemize}
do not change the payoff vector.
\end{lemma}
\begin{proof}
In order to demonstrate this property, it is important to observe that the symmetries $\sin(\phi\pm\pi)=-\sin(\phi)$ and $\cos(\phi\pm\pi)=-\cos(\phi)$ cause formula (\ref{generalEWLpayoff}) to remain unchanged for any of this substitutions.
\end{proof}
\begin{corollary}
Unitary strategies $U(\theta, \alpha, \beta)$ and $U(\theta, \alpha\pm \pi, \beta \pm \pi)$ are payoff equivalent. 
\end{corollary}
  The research \cite{frackiewicz_strong_2016} demonstrated that implementing the $\mathsf{SU}(2)$ strategies within the EWL scheme guarantees the quantum model's invariance under isomorphic transformations in the classical $2\times 2$ game. The proof of this property enables the formulation of a proposition defining a transformation that relates EWL games corresponding to isomorphic $2\times 2$ games. 
\begin{proposition}\label{propstarapraca}
Let $\Gamma$ be a $2\times 2$ game, and let $\Gamma^i$ be its isomorphic counterparts (\ref{gamma1})--(\ref{gamma3}). Let $\Gamma_{EWL} = (\{1,2\}, \{S_{1}, S_{2}\}, \{u_{1}, u_{2}\})$ be the EWL scheme for $\Gamma$ and $\Gamma^i_{EWL} = (\{1,2\}, \{S_{1}, S_{2}\}, \{u^i_{1}, u^i_{2}\})$ be the EWL schemes for $\Gamma^i$, where $S_{1} = S_{2} = \mathsf{SU}(2)$. If $\varphi\colon SU(2) \to SU(2)$ is a bijection given by
\begin{equation}\label{bijekcja}
\varphi(U(\theta, \alpha, \beta)) = U(\pi-\theta, 2\pi - \beta, \pi - \alpha),
\end{equation}
then 
\begin{align}\label{starapraca}
(u_{1}, u_{2})(U_{1}, U_{2})  = (u^1_{1}, u^1_{2})(\varphi(U_{1}), U_{2}) =(u^2_{1}, u^2_{2})(U_{1}, \varphi(U_{2})) = (u^3_{1}, u^3_{2})(\varphi(U_{1}), \varphi(U_{2})).
\end{align}
\end{proposition}
\begin{proof}
Without loss of generality, we show the first equality of (\ref{starapraca}). To find a general formula for $(u^1_{1}, u^1_{2})(U_{1}, U_{2})$, we determine the payoff function (\ref{generalEWLpayoff}) for bimatrix (\ref{gamma1}). It is equivalent to replacing $\Delta_{00}, \Delta_{01}, \Delta_{10}$ and  $\Delta_{11}$ with $\Delta_{10}, \Delta_{11}, \Delta_{00}$ and  $\Delta_{01}$, respectively. Then, substituting a strategy profile $(\varphi(U_{1}), U_{2})$ to $(u^1_{1}, u^1_{2})$ we obtain the desired equality. 
\end{proof}
Proposition~\ref{propstarapraca} clearly shows why the EWL scheme with $\mathsf{SU}(2)$ is invariant under isomorphic transformations of the classical game. In that case, the image of the transformation $\varphi$ is always an element of $\mathsf{SU}(2)$, i.e., 
\begin{equation}\label{image}
\varphi(U) \in \mathsf{SU}(2) \quad \text{for each} \quad U \in \mathsf{SU}(2). 
\end{equation}
This property does not hold for many subsets of $\mathsf{SU}(2)$. In particular, the set of strategies 
\begin{equation}
\{U(\theta, \alpha, 0)\colon \theta \in [0,\pi], \alpha \in [0,2\pi)\},
\end{equation}
commonly considered in the literature \cite{du_experimental_2002,chen_quantum_2003,li_quantum_2012,nawaz_strategic_2013,naskar_quantum_2021,anand_solving_2020}, is not closed under $\varphi$. In consequence, a pair of classical games that are isomorphic (i.e., a pair of games that are indistinguishable from the game theory point of view) may imply completely different quantum extensions of the game according to the EWL scheme. For example, it was shown in \cite{frackiewicz_strong_2016} that for games that differ only in the numbering of strategies of one of the players, the corresponding EWL games exhibited completely different Nash equilibria. In other words, the lack of invariance under isomorphic transformations of the classical game causes the EWL quantum extension to depend on the way in which the players' strategies are numbered.
On the other hand, a different two-parameter description of the quantum strategy space, as referenced e.g. in \cite{frackiewicz_nash_2022}, remains invariant under isomorphic transformations of the classical game. Therefore, the problem of determining the classes of unitary operators that preserve the invariance of the EWL scheme is extremely important for quantum game theory.

It turns out that property (\ref{image}) can be effectively utilized in quantum games with a finite number of unitary strategies. Invariance of the EWL scheme with a finite strategy set $S$ under isomorphic transformation can be guaranteed by assumption that $S$ is closed under $\varphi$. In fact, we can weaken this requirement since a given unitary strategy determines a class of payoff equivalent strategies (see also Examples \ref{ex2} and \ref{ex3}). As a result, for $U_{i} \in S = \{U_{1}, U_{2}, \dots, U_{n}\}$ we require that $\varphi(U_{i}) \in [U_{j}]$, where $U_{j} \in S$ rather than $\varphi(U_{i}) \in S$. It is stated by the following proposition:
\begin{proposition}\label{mainproposition}
Let there be a game $\Gamma$ and its isomorphic counterpart $\Gamma^i$, $i=1,2,3$. Let $\Gamma_{EWL}$ and $\Gamma_{EWL}^i$ be the EWL models for $\Gamma$ and $\Gamma^i$ respectively, in which the set of unitary operations for each player is $S = \{U_{1}, U_{2}, \dots, U_{m}\}$. Then, if 
\begin{equation}\label{zilorazowy}
\{[U_{j}]\colon U_{j} \in S\} = \{[\varphi(U_{j})]\colon U_{j} \in S\}
\end{equation}
the games 
$\Gamma_{EWL}$, $\Gamma_{EWL}^i$ are isomorphic. 
\end{proposition}
\begin{proof}
Consider $\Gamma_{EWL}$ with a strategy set $S = \{U_{1}, U_{2}, \dots, U_{m}\}$ and payoff functions $(u_{1}, u_{2})$. Note that any game $\Gamma'_{EWL}$ with payoff functions $(u_{1}, u_{2})$ and strategy set  $\{V_{1}, V_{2}, \dots, V_{m}\}$ such that $V_{j} \in [U_{j}]$ is isomorphic to $\Gamma_{EWL}$. Further, from condition (\ref{zilorazowy}) it follows that for each $j\in \{1,2, \dots, m\}$ there exist exactly one $k\in \{1,2, \dots, m\}$ such that $\varphi(U_{j}) \in [U_{k}]$. As a result, games $\Gamma_{EWL}$ with strategy sets $S$ i $S'= \{\varphi(U_{1}), \varphi(U_{2}), \dots, \varphi(U_{m}) \}$ are isomorphic. 

From the first equality of formula (\ref{starapraca}), it follows that $\Gamma_{EWL}$ with the strategy set $S$ is isomorphic to $\Gamma^1_{EWL}$ with the strategy set $S'$. Therefore the games $\Gamma_{EWL}$ and $\Gamma^1_{EWL}$ with the strategy set $S'$ are isomorphic. From condition (\ref{zilorazowy}), it follows that $S'$ forms a set of strategies equivalent in terms of payoffs with the elements of the set $S$. This means that the games $\Gamma_{EWL}$ and $\Gamma^1_{EWL}$ with the strategy set $S$ are also isomorphic. The remaining cases are shown analogously. 
\end{proof}
In what follows, we apply Proposition~\ref{mainproposition} in the construction of the EWL game with a three-element set of strategies. The problem of the EWL game with one fixed unitary strategy $U$ was considered in \cite{frackiewicz_permissible_2024}. We determined a special system of equations that enabled us to obtain $U$, and as a result to construct the EWL game that is invariant under isomorphic transformations of the classical game. The system of equations was determined by considering the EWL scheme for all isomorphic counterparts of a given $2\times 2$ game. The equations can also be obtained by directly applying (\ref{zilorazowy}). Indeed, assuming the set of strategies $S=\{I, iX, U\}$, where $U \in \mathsf{SU}(2)$ is a fixed unitary operator, we get
\begin{equation}\label{3siloraz}
\{[I], [iX], [U]\} = \{[\varphi(I)], [\varphi(iX)], [\varphi(U)]\}. 
\end{equation}
Note that for $I = U(0,0,0)$ we have $\varphi(I) = U(\pi, 0,\pi) = -U(\pi, 0,0) = -iX \in [iX]$. Similarly, $\varphi(iX) \in [I]$. As a result, Eq.~(\ref{3siloraz}) holds if $U \in [\varphi(U)]$. From (\ref{bijekcja}) it follows that $U(\theta, \alpha, \beta) \in [U(\pi-\theta, 2\pi - \beta, \pi - \alpha)]$. This, in turn, means that the unitary operation $U(\theta, \alpha, \beta)$ is payoff equivalent to $U(\pi-\theta, 2\pi - \beta, \pi - \alpha)$. This implies equations 
\begin{equation}\label{rownanieS}
u(U(\pi-\theta, 2\pi - \beta, \pi - \alpha), s) = u(U(\theta, \alpha, \beta)), s)~\text{for}~s\in \{I, iX, U(\theta, \alpha, \beta)\}.
\end{equation}
From (\ref{rownanieS}), we obtain exactly the same equations as (48)-(50) in \cite{frackiewicz_permissible_2024}.
\section{Determining the criteria for valid unitary matrices $U_1$ and $U_2$}\label{determining_the_criteria}
In this section, we use Proposition~\ref{mainproposition} to determine unitary strategies $U_1$ and $U_2$ such that the EWL game with a finite set of strategies $\{I, iX, U_{1}, U_{2}\}$ remains invariant under isomorphic transformations of the classical game. Since $\varphi(I) \in [iX]$ and $\varphi(iX) \in [I]$, Eq.~(\ref{zilorazowy}) is satisfied provided two conditions are met
\begin{enumerate}
\item $\varphi(U_{1}) \in [U_{2}]$ and $\varphi(U_{2}) \in [U_{1}],$
\item $\varphi(U_{1}) \in [U_{1}]$ and $\varphi(U_{2}) \in [U_{2}].$ 
\end{enumerate}
Let us first consider case 1. According to Definition~\ref{meyerson}, the condition $\varphi(U_{1}) \in [U_{2}]$ leads to the following equations:
\begin{align}
u(U_{1}(\pi - \theta_1, 2\pi - \beta_1, \pi - \alpha_1), I) &= u(U_{2}(\theta_2, \alpha_2, \beta_2), I) \label{cnd1},\\
u(U_{1}(\pi - \theta_1, 2\pi - \beta_1, \pi - \alpha_1), iX) &= u(U_{2}(\theta_2, \alpha_2, \beta_2), iX) \label{cnd2}, \\
u(U_{1}(\pi - \theta_1, 2\pi - \beta_1, \pi - \alpha_1), U_{1}(\theta_1, \alpha_1, \beta_1)) &= u(U_{2}(\theta_2, \alpha_2, \beta_2), U_{1}(\theta_1, \alpha_1, \beta_1)), \label{cnd3} \\
u(U_{1}(\pi - \theta_1, 2\pi - \beta_1, \pi - \alpha_1), U_{2}(\theta_2, \alpha_2, \beta_2)) &= u(U_{2}(\theta_2, \alpha_2, \beta_2), U_{2}(\theta_2, \alpha_2, \beta_2)). \label{cnd4}
\end{align}
Similarly, condition $\varphi(U_{2}) \in [U_{1}]$ leads to:
\begin{align}
u(U_{2}(\pi - \theta_2, 2\pi - \beta_2, \pi - \alpha_2), I) &= u(U_{1}(\theta_1, \alpha_1, \beta_1), I), \label{cnd5}\\
u(U_{2}(\pi - \theta_2, 2\pi - \beta_2, \pi - \alpha_2), iX) &= u(U_{1}(\theta_1, \alpha_1, \beta_1), iX), \label{cnd6}\\
u(U_{2}(\pi - \theta_2, 2\pi - \beta_2, \pi - \alpha_2), U_{1}(\theta_1, \alpha_1, \beta_1)) &= u(U_{1}(\theta_1, \alpha_1, \beta_1), U_{1}(\theta_1, \alpha_1, \beta_1)), \label{cnd7}\\
u(U_{2}(\pi - \theta_2, 2\pi - \beta_2, \pi - \alpha_2), U_{2}(\theta_2, \alpha_2, \beta_2)) &= u(U_{1}(\theta_1, \alpha_1, \beta_1), U_{2}(\theta_2, \alpha_2, \beta_2)). \label{cnd8}
\end{align}
Formulas (\ref{cnd1}), (\ref{cnd2}), (\ref{cnd5}) and (\ref{cnd6}) lead to three conditions relating the parameters of the first strategy to the second (or to their equivalent forms):
\begin{align}
   \sin^2\frac{\theta_2}{2} &= \cos^2\frac{\theta_1}{2}, \label{theta12} \\
    \sin^2\alpha_2 &= \sin^2\beta_1, \label{a2b1}  \\
    \sin^2\beta_2 &= \sin^2\alpha_1. \label{a1b2}
\end{align}
Note that (\ref{theta12})-(\ref{a1b2}) are respectively equivalent to 
\begin{align}
 \theta_2 &= \pi - \theta_1,  \label{theta12s} \\
 \alpha_2 &= n \pi \pm \beta_1 \label{a2b1s},\\
 \beta_2 &= l \pi \pm \alpha_1 \label{a1b2s},
\end{align}
where $\theta_1,\theta_2 \in [0,\pi]$ and $\alpha_1,\beta_1, \alpha_2,\beta_2 \in [0,2\pi)$ and $n, l \in \mathbb{Z}$. The following section will explore the repercussions of the remaining equations (\ref{cnd3}), (\ref{cnd4}), (\ref{cnd7}) and (\ref{cnd8}), and how they may fine-tune the conditions (\ref{theta12s})-(\ref{a1b2s}) in special cases. 

In the event that $\theta_1=0$, $\theta_2=\pi$, these equations reduce to simple relationships
\begin{equation}\label{sin2b}   \sin^2(2\beta_2)=\sin^2(2\alpha_1)=\sin^2(\alpha_1-\beta_2),
\end{equation}
which are met by $\alpha_1+\beta_2= n \pi$, $n \in \mathbb{Z}$, i.e., particular solutions of the equation (\ref{a1b2s}), $\alpha_2, \beta_1 \in [0, 2\pi)$ take arbitrary values in this case. In case when $\theta_1=\pi$, $\theta_2=0$, the parameters must obey $\alpha_2+\beta_1= n \pi$, $n \in \mathbb{Z}$ - the  particular form of (\ref{a2b1s}), and $\alpha_1, \beta_2 \in [0, 2\pi)$ are arbitrary. 

In continued analysis, we will proceed with the assumption that $0< \theta_2=\pi - \theta_1<\pi$. In this case (\ref{cnd3}, \ref{cnd4}, \ref{cnd7}, \ref{cnd8}) lead to a system of 16 equations
\begin{align}
&\left(\cos(\alpha_1-\beta_1)+\sin(\alpha_1-\beta_1)   \right)^2 =  \left(\cos(\alpha_1+\alpha_2)+\sin(\beta_1+\beta_2)\right)^2, \label{eq1}\\
&\left(\sin(\alpha_1-\beta_1)-\cos(\alpha_1-\beta_1)\right)^2 
=\left(\sin(\alpha_1+\alpha_2)-\cos(\beta_1+\beta_2)\right)^2, \\ 
&\left( \cos(\beta_1+\beta_2)-\sin(\alpha_1+\alpha_2)\right)^2 =\left( \cos(\alpha_2-\beta_2)+\sin(\alpha_2-\beta_2)\right)^2, \\ 
&\left( \sin(\beta_1+\beta_2)-\cos(\alpha_1+\alpha_2)\right)^2 =\left( \sin(\alpha_2-\beta_2)+\cos(\alpha_2-\beta_2)\right)^2, \\ 
&\left(\cos(\beta_1+\beta_2)-\sin(\alpha_1+\alpha_2)\right)^2  = \left( \cos(\alpha_1-\beta_1)+\sin(\alpha_1-\beta_1)\right)^2, \\ 
&\left( \sin(\beta_1+\beta_2)-\cos(\alpha_1+\alpha_2)\right)^2 = \left( \sin(\alpha_1-\beta_1)+\cos(\alpha_1-\beta_1)\right)^2, \\ 
& \left(\cos(\beta_2-\alpha_2)+\sin(\alpha_2-\beta_2)\right)^2  = \left(\cos(\alpha_1+\alpha_2)+\sin(\beta_1+\beta_2)\right)^2, \\
&\left(\sin(\beta_2-\alpha_2)+\cos(\alpha_2-\beta_2)\right)^2 = \left(\sin(\alpha_1+\alpha_2)-\cos(\beta_1+\beta_2)\right)^2, \\
  &\left( \cos(2\beta_1)\sin^2\frac{\theta_1}{2}-\sin(2\alpha_1)\cos^2\frac{\theta_1}{2}\right)^2
  = \left( \cos(\alpha_2-\beta_1)\sin^2\frac{\theta_1}{2}+\sin(\alpha_1-\beta_2)\cos^2\frac{\theta_1}{2}\right)^2, \label{eq9} \\  
  &\left( \sin(2\beta_1)\sin^2\frac{\theta_1}{2}-\cos(2\alpha_1)\cos^2\frac{\theta_1}{2}\right)^2 = \left( \sin(\alpha_2-\beta_1)\sin^2\frac{\theta_1}{2}+\cos(\alpha_1-\beta_2)\cos^2\frac{\theta_1}{2}\right)^2, \label{eq10} \\ 
&\left( \cos(2\beta_2)\cos^2\frac{\theta_1}{2}-\sin(2\alpha_2)\sin^2\frac{\theta_1}{2}\right)^2 = \left( \cos(\alpha_1-\beta_2)\cos^2\frac{\theta_1}{2}+\sin(\alpha_2-\beta_1)\sin^2\frac{\theta_1}{2}\right)^2, \\ 
 &\left( \sin(2\beta_2)\cos^2\frac{\theta_1}{2}-\cos(2\alpha_2)\sin^2\frac{\theta_1}{2}\right)^2 = \left( \sin(\alpha_1-\beta_2)\cos^2\frac{\theta_1}{2}+\cos(\alpha_2-\beta_1)\sin^2\frac{\theta_1}{2}\right)^2, \\ 
&\left(\cos(\beta_1-\alpha_2)\sin^2\frac{\theta_1}{2}+\sin(\alpha_1-\beta_2)\cos^2\frac{\theta_1}{2}\right)^2
    = \left(\cos(2\alpha_2)\sin^2\frac{\theta_1}{2}+\sin(2\beta_2)\cos^2\frac{\theta_1}{2}\right)^2, \\ 
 &\left(\sin(\beta_1-\alpha_2)\sin^2\frac{\theta_1}{2}+\cos(\alpha_1-\beta_2)\cos^2\frac{\theta_1}{2}\right)^2  = \left(\sin(2\alpha_2)\sin^2\frac{\theta_1}{2}-\cos(2\beta_2)\cos^2\frac{\theta_1}{2}\right)^2, \\ 
 &\left(\cos(\beta_2-\alpha_1)\cos^2\frac{\theta_1}{2}+\sin(\alpha_2-\beta_1)\sin^2\frac{\theta_1}{2}\right)^2 = \left(\cos(2\alpha_1)\cos^2\frac{\theta_1}{2}+\sin(2\beta_1)\sin^2\frac{\theta_1}{2}\right)^2,  \\ 
& \left(\sin(\beta_2-\alpha_1)\cos^2\frac{\theta_1}{2}+\cos(\alpha_2-\beta_1)\sin^2\frac{\theta_1}{2}\right)^2  = \left(\sin(2\alpha_1)\cos^2\frac{\theta_1}{2}-\cos(2\beta_1)\sin^2\frac{\theta_1}{2}\right)^2. \label{eq16} 
\end{align}
It can be deduced from this system of equations and formulas (\ref{a2b1s}) and (\ref{a1b2s}) that
 \begin{align}\label{twoeq}
 &\sin(2\beta_1)\cos(2\alpha_1)= 0, \nonumber \\
 &\sin2(\alpha_1-\beta_1) = 0.
 \end{align}
The last equations lead to two distinct sets of solutions
\begin{align}
&(\alpha_{1}, \beta_{1}) \in \left\{\frac{\pi}{4}, \frac{3\pi}{4}, \frac{5\pi}{4}, \frac{7\pi}{4}\right\} \times \left\{\frac{\pi}{4}, \frac{3\pi}{4}, \frac{5\pi}{4}, \frac{7\pi}{4}\right\}, \label{catBC}\\
&(\alpha_{1}, \beta_{1}) \in \left\{0, \frac{\pi}{2}, \pi, \frac{3\pi}{2}\right\}\times \left\{0, \frac{\pi}{2}, \pi, \frac{3\pi}{2}\right\}, \label{catDE}
\end{align}
with 16 combinations of $\alpha_{1}$ and $\beta_{1}$ in each.

If the parameters correspond to category (\ref{catBC}), then $\alpha_2$ and $\beta_2$ calculated using Eqs (\ref{a2b1s}) and (\ref{a1b2s}), will fall within the same Cartesian product (\ref{catBC}). The system of equations (\ref{eq9})-(\ref{eq16}) contains arguments of the type $(\alpha_1-\beta_2)$ or $(\alpha_2-\beta_1)$. There are four potential cases for their differences: 
\begin{enumerate}\label{fourscenarios}
\item $\alpha_{2} - \beta_{1} = n\pi$, $\alpha_{1} - \beta_{2} = l\pi$ 
\item $\alpha_{2} - \beta_{1} = n\pi$, $\alpha_{1} - \beta_{2} = l\pi+ \frac{\pi}{2}$ 
\item $\alpha_{2} - \beta_{1} = n\pi + \frac{\pi}{2}$, $\alpha_{1} - \beta_{2} = l\pi$ 
\item $\alpha_{2} - \beta_{1} = n\pi+ \frac{\pi}{2}$, $\alpha_{1} - \beta_{2} = l\pi+ \frac{\pi}{2}$, 
\end{enumerate}
where $n$, $l$ are integers. Both cases 2 and 3 can be disqualified as they do not satisfy Eq. (\ref{eq10}) and Eq. (\ref{eq9}), respectively. When condition 1 is assumed, equations (\ref{eq9})-(\ref{eq16}) will be fulfilled provided
\begin{equation}\label{cos4}
    \cos^4\frac{\theta_1}{2}=\sin^4\frac{\theta_1}{2},
\end{equation}
therefore $\theta_{1} = \theta_{2} =\frac{\pi}{2}$. The set of all such solutions is in this case
\begin{align} \label{solB}
\left\{(\alpha_{1}, \beta_{1}, \alpha_{2}, \beta_{2})\in \left\{\frac{\pi}{4}, \frac{3\pi}{4}, \frac{5\pi}{4}, \frac{7\pi}{4}\right\}^4 \colon \alpha_{2} = \beta_{1} + n\pi \wedge \beta_{2} = \alpha_{1} + l\pi \wedge n,l \in \mathbb{Z} \right\},
\end{align}
the set (\ref{solB}) have $4\times4\times2\times2=64$ elements and can be written explicitly as
\begin{align}\label{solBB}
&\left\{\frac{\pi}{4}, \frac{5\pi}{4}\right\} \times\left\{\frac{\pi}{4}, \frac{5\pi}{4}\right\} \times\left\{\frac{\pi}{4}, \frac{5\pi}{4}\right\} \times\left\{\frac{\pi}{4}, \frac{5\pi}{4}\right\}
\cup \left\{\frac{\pi}{4}, \frac{5\pi}{4}\right\} \times\left\{\frac{3\pi}{4}, \frac{7\pi}{4}\right\} \times\left\{\frac{3\pi}{4}, \frac{7\pi}{4}\right\} \times\left\{\frac{\pi}{4}, \frac{5\pi}{4}\right\} \cup \\ \nonumber
&\left\{\frac{3\pi}{4}, \frac{7\pi}{4}\right\} \times\left\{\frac{\pi}{4}, \frac{5\pi}{4}\right\} \times\left\{\frac{\pi}{4}, \frac{5\pi}{4}\right\} \times\left\{\frac{3\pi}{4}, \frac{7\pi}{4}\right\}
\cup \left\{\frac{3\pi}{4}, \frac{7\pi}{4}\right\} \times\left\{\frac{3\pi}{4}, \frac{7\pi}{4}\right\} \times\left\{\frac{3\pi}{4}, \frac{7\pi}{4}\right\} \times\left\{\frac{3\pi}{4}, \frac{7\pi}{4}\right\}.
\end{align}
If condition 4 is met, all the equations (\ref{eq1}-\ref{eq16}) are satisfied for all $\theta \in (0, \pi)$ and the remaining solutions are of the form
\begin{align} \label{solC}
\left\{(\alpha_{1}, \beta_{1}, \alpha_{2}, \beta_{2})\in \left\{\frac{\pi}{4}, \frac{3\pi}{4}, \frac{5\pi}{4}, \frac{7\pi}{4}\right\}^4 \colon \alpha_{2} = \beta_{1} + (n+\frac{1}{2})\pi \wedge \beta_{2} = \alpha_{1} + (l+\frac{1}{2})\pi \wedge n,l \in \mathbb{Z} \right\},
\end{align}
the set (\ref{solC}) have also $4\times4\times2\times2=64$ elements and can be written explicitly as
\begin{align}\label{solCC}
&\left\{\frac{\pi}{4}, \frac{5\pi}{4}\right\} \times\left\{\frac{\pi}{4}, \frac{5\pi}{4}\right\} \times\left\{\frac{3\pi}{4}, \frac{7\pi}{4}\right\} \times\left\{\frac{3\pi}{4}, \frac{7\pi}{4}\right\}
\cup \left\{\frac{3\pi}{4}, \frac{7\pi}{4}\right\} \times\left\{\frac{3\pi}{4}, \frac{7\pi}{4}\right\} \times\left\{\frac{
\pi}{4}, \frac{5\pi}{4}\right\} \times\left\{\frac{\pi}{4}, \frac{5\pi}{4}\right\} \cup \\ \nonumber
&\left\{\frac{\pi}{4}, \frac{5\pi}{4}\right\} \times\left\{\frac{3\pi}{4}, \frac{7\pi}{4}\right\} \times\left\{\frac{\pi}{4}, \frac{5\pi}{4}\right\} \times\left\{\frac{3\pi}{4}, \frac{7\pi}{4}\right\}
\cup \left\{\frac{3\pi}{4}, \frac{7\pi}{4}\right\} \times\left\{\frac{\pi}{4}, \frac{5\pi}{4}\right\} \times\left\{\frac{3\pi}{4}, \frac{7\pi}{4}\right\} \times\left\{\frac{\pi}{4}, \frac{5\pi}{4}\right\}.
\end{align}
In the event, that the parameters pertain to category (\ref{catDE}), $\alpha_2$ and $\beta_2$ calculated through Eqs (\ref{a2b1s}) and (\ref{a1b2s}) again belong to the same Cartesian product (\ref{catDE}). In this case all the equations (\ref{eq1})-(\ref{eq16}) are fully met for all $\theta \in (0, \pi)$ without any additional conditions. For this class of solutions we will introduce an additional distinction due to the relationship between parameters $\alpha_1$ and $\beta_1$, because they will lead to different payoff matrices. The first class is defined by relation $\beta_1 = \alpha_1 + n \pi$, where $n \in \mathbb{Z}$
\begin{align} \label{solD}
\left\{(\alpha_{1}, \beta_{1}, \alpha_{2}, \beta_{2})\in \left\{0, \frac{\pi}{2}, \pi, \frac{3\pi}{2}\right\}^4 \colon \beta_{1} = \alpha_{1} + n\pi \wedge \alpha_{2} = \beta_{1} + l\pi \wedge \beta_{2} = \alpha_{1} + m\pi \wedge n,l,m \in \mathbb{Z} \right\},
\end{align}
the set of all solutions corresponding to the class (\ref{solD}) have $4\times 2\times2 \times2 =32$ elements is equal to
\begin{align}\label{solDD}
\left\{0, \pi\right\}\times \left\{0, \pi\right\} \times \left\{0, \pi\right\}\times \left\{0, \pi\right\} \cup \left\{\frac{\pi}{2}, \frac{3\pi}{2}\right\} \times\left\{\frac{\pi}{2}, \frac{3\pi}{2}\right\} \times\left\{\frac{\pi}{2}, \frac{3\pi}{2}\right\} \times\left\{\frac{\pi}{2}, \frac{3\pi}{2}\right\}.
\end{align}
The second class corresponding to the condition $\beta_1=\alpha_1 + (n+\frac{1}{2}) \pi$, $n \in \mathbb{Z}$, is
\begin{align} \label{solE}
\left\{(\alpha_{1}, \beta_{1}, \alpha_{2}, \beta_{2})\in \left\{0, \frac{\pi}{2}, \pi, \frac{3\pi}{2}\right\}^4 \colon \beta_{1} = \alpha_{1} + (n+\frac{1}{2})\pi \wedge \alpha_{2} = \beta_{1} + l\pi \wedge \beta_{2} = \alpha_{1} + m\pi \wedge n,l,m \in \mathbb{Z} \right\},
\end{align}
it also has $4\times 2\times2 \times2 =32$ elements and the set (\ref{solE}) is equal to
\begin{align}\label{solEE}
\left\{0, \pi \right\}\times \left\{\frac{\pi}{2}, \frac{3\pi}{2}\right\} \times \left\{\frac{\pi}{2}, \frac{3\pi}{2}\right\}\times \left\{0, \pi \right\} \cup \left\{\frac{\pi}{2}, \frac{3\pi}{2}\right\} \times\left\{0, \pi \right\} \times \left\{0, \pi \right\}\times \left\{\frac{\pi}{2}, \frac{3\pi}{2}\right\}.
\end{align}
The second case $\varphi(U_{j}) \in [U_{j}]$ for $j =1,2$ is analogous to one considered in \cite{frackiewicz_permissible_2024}. We obtain equations in the form
\begin{equation}\label{st1}
u(U_{1}(\pi-\theta_{1}, 2\pi - \beta_{1}, \pi - \alpha_{1}), s) = u(U_{1}(\theta_{1}, \alpha_{1}, \beta_{1})), s), ~s\in \{I, iX, U_{1}(\theta_{1}, \alpha_{1}, \beta_{1}), U_{2}(\theta_{2}, \alpha_{2}, \beta_{2})\}
\end{equation}
and
\begin{equation}\label{st2}
u(U_{2}(\pi-\theta_{2}, 2\pi - \beta_{2}, \pi - \alpha_{2}), s) = u(U_{2}(\theta_{2}, \alpha_{2}, \beta_{2})), s), ~s\in \{I, iX, U_{1}(\theta_{1}, \alpha_{1}, \beta_{1}), U_{2}(\theta_{2}, \alpha_{2}, \beta_{2})\}.
\end{equation}
Note that eq.~(\ref{st1}) for $s\in \{I, iX, U_{1}\}$ and eq.~(\ref{st2}) for $s\in \{I, iX, U_{2}\}$ have the same form as the equations (48)-(50) in \cite{frackiewicz_permissible_2024}. As a result solutions $(\theta_{1}, \alpha_{1}, \beta_{1})$ of (\ref{st1}) and $(\theta_{2}, \alpha_{2}, \beta_{2})$ of (\ref{st2}) belong to 3 classes defined by equations (78)-(80) of \cite{frackiewicz_permissible_2024}. In the remaining cases, where $s=U_{2}(\theta_{2}, \alpha_{2}, \beta_{2})$ in eq.~(\ref{st1}) and $s = U_{1}(\theta_{1}, \alpha_{1}, \beta_{1})$ in eq.~(\ref{st2}) the solutions are particular cases of solutions from the first case i.e., $\varphi(U_{1}) \in [U_{2}]$ and $\varphi(U_{2}) \in [U_{1}]$.
\section{Permissible extensions of classical $2\times 2$ games combining four strategies}
In the previous section, five types of solutions obeying the criteria of invariance with respect to isomorphic transformation of the game defined by the bimatrix (\ref{bimatrixiso}) were found. Here we will define permissible game extensions, defined by their payoff matrices, corresponding to these solutions.  Notwithstanding, prior to providing particular matrices, we establish a lemma that ensures the invariance of the expansion matrices showcased in the subsequent sections with regards to isomorphic transformations of the initial game. It is important to note that
\begin{lemma}\label{lemmaH}
  Any matrix $M$ presented in the form of linear matrix combinations of $\Gamma^0, \Gamma^1, \Gamma^2, \Gamma^3$,
\begingroup 
\setlength\arraycolsep{4.5pt}
\renewcommand\arraystretch{1.5}
\begin{equation}\label{lemma2}
M = \begin{pmatrix}
\sum_{i=0}^3 e_i\Gamma^i & \sum_{i=0}^3 f_i\Gamma^i \\ 
\sum_{i=0}^3 g_i\Gamma^i & \sum_{i=0}^3 h_i\Gamma^i
\end{pmatrix}, \quad \text{where} \quad e_i, f_i, g_i, h_i\in \mathrm{R},
\end{equation}
\endgroup
remains invariant with respect to any isomorphic transformation of the initial matrix game $\Gamma=\Gamma^0$.
\end{lemma}
\begin{proof}
Replacing rows in the initial matrix $\Gamma^0$, i.e. assuming that $\Delta'_{00}=\Delta_{10}$, $\Delta'_{01}=\Delta_{11}$, $\Delta'_{10}=\Delta_{00}$ and $\Delta'_{11}=\Delta_{01}$, leads to substitutions $(\Gamma^0)'=\Gamma^1$, $(\Gamma^{1})'=\Gamma^0$, $(\Gamma^{2})'=\Gamma^3$ i $(\Gamma^{3})'=\Gamma^2$, and consequently  
\begingroup 
\setlength\arraycolsep{4.5pt}
\renewcommand\arraystretch{1.5}
\begin{equation}\label{matrixL3}
\medmuskip = 0.2mu
M' = \begin{pmatrix}
\sum_{i=0}^3 e_i(\Gamma^i)' & \sum_{i=0}^3 f_i(\Gamma^i)' \\ 
\sum_{i=0}^3 g_i(\Gamma^i)' & \sum_{i=0}^3 h_i(\Gamma^i)'
\end{pmatrix} =
 \begin{pmatrix}
e_0\Gamma^1+e_1\Gamma^0+e_2\Gamma^3+e_3\Gamma^2 & f_0\Gamma^1+f_1\Gamma^0+f_2\Gamma^3+f_3\Gamma^2 \\ 
g_0\Gamma^1+g_1\Gamma^0+g_2\Gamma^3+g_3\Gamma^2 & h_0\Gamma^1+h_1\Gamma^0+h_2\Gamma^3+h_3\Gamma^2
\end{pmatrix}.
\end{equation}
\endgroup
Matrix $M'$ is isomorphic with $M$, and the transformation establishing this isomorphism is the replacement of rows 1 with 2 and 3 with 4 in the matrix on the right side of the equation (\ref{matrixL3}). For the remaining isomorphic transformations of the $\Gamma^0$ matrix, the proof is analogous, with the replacement of columns leading to substitutions $(\Gamma^0)'=\Gamma^2$, $(\Gamma^{1})'=\Gamma^3$, $(\Gamma^{2})'=\Gamma^0$ i $(\Gamma^{3})'=\Gamma^1$ and replacement of rows and columns leading to substitutions $(\Gamma^0)'=\Gamma^3$, $(\Gamma^{1})'=\Gamma^2$, $(\Gamma^{2})'=\Gamma^1$ i $(\Gamma^{3})'=\Gamma^0$.
\end{proof}
It should be pointed out that Lemma \ref{lemmaH} holds true for all extensions where the number of players' strategies is even and the game matrix follows a structure similar to (\ref{lemma2}). The demonstrated lemma guarantees that all five forms (A-E) of the classical game extensions outlined below are inherently unaffected by isomorphic transformations of the original game, as long as they are presented in the form (\ref{lemma2}).
\subsection{Extension of the A class}
The first type's extension corresponds to $\{\theta_1, \theta_2\} = \{0,\pi\}$ and remaining parameters satisfying Eq. (\ref{sin2b}). For this situation, it is necessary for the parameters $\alpha_1,\beta_1, \alpha_2,\beta_2 \in [0,2\pi)$ to fulfill one of the conditions
\begin{align}
&\alpha_1+\beta_2=n \pi \hspace{10pt}\text{if} \hspace{10pt} \theta_1=0, \,\theta_2=\pi \hspace{10pt}\text{or} \label{t10}\\
&\alpha_2+\beta_1=n \pi \hspace{10pt}\text{if} \hspace{10pt} \theta_1=\pi, \, \theta_2=0, \label{t1pi}
\end{align}
where $n \in \mathbb{Z}$. Matrices of the extended game, corresponding to (\ref{t10}) and (\ref{t1pi}) are:
\begingroup 
\setlength\arraycolsep{4.5pt}
\renewcommand\arraystretch{1.5}
\begin{equation} \label{klasaA}
\medmuskip = 0.2mu
A_1 = \begin{pmatrix}
\Gamma & a_1\Gamma+a_1'\Gamma^3 \\ 
a_1\Gamma+a_1'\Gamma^3 & b_1\Gamma+b_1'\Gamma^3
\end{pmatrix}, \quad 
A_2 = \begin{pmatrix}
\Gamma & a_2\Gamma^2+a_2'\Gamma^1 \\ 
a_2\Gamma^1+a_2'\Gamma^2 & b_2\Gamma^3+b_2'\Gamma
\end{pmatrix},
\end{equation}
\endgroup
respectively, where $a_i=\cos^{2}\alpha_i$, $a_i'= 1-a_i =\sin^{2}\alpha_i$ and $b_i=\cos^{2}2\alpha_i$, $b_i'= 1-b_i =\sin^{2}2\alpha_i$, $i=1,2$. Significant scenarios for this type of extension include for $A_1$: 
\begingroup 
\setlength\arraycolsep{4.5pt}
\renewcommand\arraystretch{1.5}
\begin{equation} \label{klasaA1p}
\medmuskip = 0.2mu
\left. A_1\right|_{\alpha_{1} =0} = \begin{pmatrix}
\Gamma & \Gamma \\ 
\Gamma & \Gamma
\end{pmatrix}, \quad 
\left. A_1\right|_{\alpha_{1} = \frac{\pi}{2}} = \begin{pmatrix}
\Gamma & \Gamma^3 \\ 
\Gamma^3 & \Gamma
\end{pmatrix},
\end{equation}
\endgroup
and for $A_2$
\begingroup 
\setlength\arraycolsep{4.5pt}
\renewcommand\arraystretch{1.5}
\begin{equation} \label{klasaA2p}
\medmuskip = 0.2mu
\left. A_2\right|_{\alpha_{2} =0} =  \begin{pmatrix}
\Gamma & \Gamma^2 \\ 
\Gamma^1 & \Gamma^3
\end{pmatrix}, \quad 
\left. A_2\right|_{\alpha_{2} = \frac{\pi}{2}} = \begin{pmatrix}
\Gamma & \Gamma^1 \\ 
\Gamma^2 & \Gamma^3
\end{pmatrix}.
\end{equation}
\endgroup
 The explicit form of $A_1$ for $\alpha_1 = \beta_2 =\frac{\pi}{2}$ is 
\begin{equation}\label{Pauli}
\bordermatrix{ & I & iX & U_{1} & U_{2} \cr 
I & \Delta_{00} & \Delta_{01} & \Delta_{11} & \Delta_{10}\cr 
iX & \Delta_{10} & \Delta_{11} & \Delta_{01} & \Delta_{00} \cr 
U_{1} & \Delta_{11} & \Delta_{10} & \Delta_{00} & \Delta_{01} \cr 
U_{2} & \Delta_{01} & \Delta_{00} & \Delta_{10} & \Delta_{11}}.
\end{equation}
Some instances of this particular type of expansion has been used previously e.g. in \cite{giannakis_quantum_2019,consuelo-leal_pareto-optimal_2019,szopa_efficiency_2021}.
Note, that the strategy matrices $U(\theta, \alpha, \beta)$ defined by (\ref{Umatrix}) are diagonal or anti-diagonal for $\theta_1 = 0$ or $\theta_1 = \pi$, respectively. In the special case of (\ref{Pauli}) the operators generating the extension are Pauli matrices:
\begin{equation}
\setlength\arraycolsep{3.5pt}
\renewcommand\arraystretch{1.0}
iX=i\sigma_x = \begin{pmatrix}
0 & i \\ 
i & 0
\end{pmatrix}, \quad U_1=i\sigma_z = \begin{pmatrix}
i & 0 \\ 
0 & -i
\end{pmatrix}, \quad U_2=i\sigma_y = \begin{pmatrix}
0 & -1 \\ 
1 & 0
\end{pmatrix}.
\end{equation}
In the remaining two cases $\left. A_2\right|_{\alpha_{2} =0}$ and $\left. A_2\right|_{\alpha_{2} =\frac{\pi}{2}}$ the explicit forms of the extensions (\ref{klasaA2p}) are
\begingroup 
\setlength\arraycolsep{4.5pt}
\renewcommand\arraystretch{1.5}
\begin{equation}\label{Pauli2}
\bordermatrix{ & I & iX & U_{1} & U_{2} \cr 
I & \Delta_{00} & \Delta_{01} & \Delta_{01} & \Delta_{00}\cr 
iX & \Delta_{10} & \Delta_{11} & \Delta_{11} & \Delta_{10} \cr 
U_{1} & \Delta_{10} & \Delta_{11} & \Delta_{11} & \Delta_{10} \cr 
U_{2} & \Delta_{00} & \Delta_{01} & \Delta_{01} & \Delta_{00}}, \quad \quad \quad 
\bordermatrix{ & I & iX & U_{1} & U_{2} \cr 
I & \Delta_{00} & \Delta_{01} & \Delta_{10} & \Delta_{11}\cr 
iX & \Delta_{10} & \Delta_{11} & \Delta_{00} & \Delta_{01} \cr 
U_{1} & \Delta_{01} & \Delta_{00} & \Delta_{11} & \Delta_{10} \cr 
U_{2} & \Delta_{11} & \Delta_{10} & \Delta_{01} & \Delta_{00}},
\end{equation}
\endgroup
where $U_1=i\sigma_x$, $U_2=I$ and $U_1=i\sigma_y$, $U_2=i\sigma_z$, respectively. 
\subsection{Extension of the B class}
 This particular extension satisfies Eq. (\ref{cos4}) and therefore $\theta_1=\theta_2=\frac{\pi}{2}$. The phase parameters $\alpha$ and $\beta$ are multiples of $\frac{\pi}{4}$ with the restrictions that $\alpha_2-\beta_1=n\pi$ and $\beta_2-\alpha_1=l\pi$, with integer $n,l$. For all these parameters (\ref{solBB}) the payoff matrix of the game is 
 \begingroup 
\setlength\arraycolsep{4.5pt}
\renewcommand\arraystretch{1.5}
\begin{equation} \label{klasaB}
B = \begin{pmatrix}
\medmuskip = 0.2mu
\Gamma & \frac{\Gamma+\Gamma^1+\Gamma^2+\Gamma^3}{4} \\ 
\frac{\Gamma+\Gamma^1+\Gamma^2+\Gamma^3}{4} & \frac{\Gamma+\Gamma^1+\Gamma^2+\Gamma^3}{4}
\end{pmatrix}, 
\end{equation}
\endgroup
or explicitly
\begin{equation}\label{klasaB1}
\bordermatrix{ & I & iX & U_{1} & U_{2} \cr 
I & \Delta_{00} & \Delta_{01} & \Delta & \Delta \cr 
iX & \Delta_{10} & \Delta_{11} & \Delta & \Delta \cr 
U_{1} & \Delta & \Delta & \Delta & \Delta \cr 
U_{2} & \Delta & \Delta & \Delta & \Delta},  \quad \text{where} \quad \Delta = \frac{\Delta_{00}+\Delta_{01}+\Delta_{10}+\Delta_{11}}{4}.
\end{equation}
An example of operators that result in a type B extension are the player strategies listed below
\begingroup 
\setlength\arraycolsep{4.5pt}
\renewcommand\arraystretch{1.5}
\begin{equation}
U_{1} = U\left(\frac{\pi}{2}, \frac{\pi}{4}, \frac{3\pi}{4}\right) = \frac{1}{2}\begin{pmatrix}
-1+i & -1-i \\ 
1-i & -1-i
\end{pmatrix}, \quad U_{2} =U\left(\frac{\pi}{2}, \frac{3\pi}{4}, \frac{\pi}{4}\right) = \frac{1}{2}\begin{pmatrix}
-1+i &-1+i \\ 
1+i & -1-i
\end{pmatrix}.
\end{equation}
\endgroup
\subsection{Extension of the C class}
 In the case of this extension, both $\theta_1$ and $\theta_2$ can vary continuously within the range of (0, $\pi$), with the condition that $\theta_2$ equals $\pi - \theta_1$. The phase parameters $\alpha$ and $\beta$ are multiples of $\frac{\pi}{4}$ with the restrictions that $\alpha_2-\beta_1=n\pi+\frac{\pi}{2}$ and $\beta_2-\alpha_1=l\pi+\frac{\pi}{2}$,  with integer $n,l$. The set of all these parameters is (\ref{solCC}). The corresponding extended game matrix is
 \begingroup 
\setlength\arraycolsep{4.5pt}
\renewcommand\arraystretch{1.5}
\begin{equation}
\medmuskip = 0.2mu
C = \begin{pmatrix} \label{klasaC}
\Gamma & t\frac{\Gamma+\Gamma^3}{2}+t'\frac{\Gamma^1+\Gamma^2}{2} \\ 
t\frac{\Gamma+\Gamma^3}{2}+t'\frac{\Gamma^1+\Gamma^2}{2} & t'^2\Gamma + tt'(\Gamma^1+\Gamma^2) + t^2\Gamma^3
\end{pmatrix}, 
\end{equation}
\endgroup
where $t=\cos^{2}\frac{\theta_1}{2}$, $t'= 1-t =\sin^{2}\frac{\theta_1}{2}$. In the special case when $\theta_1=\theta_2=\frac{\pi}{2}$, the C class extension reduces to a form (\ref{klasaB}) of class B, which then applies to both types of parameters (\ref{solBB}) and (\ref{solCC}).

To illustrate that adding two unitary operations can significantly impact the course and consequently the final outcome of the game, we use a well-known problem in game theory called the Prisoner's Dilemma (PD). The typical matrix representation of this game can be expressed as a bimatrix
\begin{equation}
\begin{pmatrix}
(3,3) & (0,5) \\ 
(5,0) & (1,1)
\end{pmatrix}.
\end{equation}
Following (\ref{klasaC}), the Prisoner's dilemma in the class C extension is of the form
\begin{equation}
\renewcommand\arraystretch{1.5}
\medmuskip = 0.6mu
C = \begin{pmatrix} 
(3,3) & 
(0,5) &
\left(\frac{5-t}{2}, \frac{5-t}{2}\right) &
\left(\frac{t+4}{2}, \frac{t+4}{2}\right) \\
(5,0) &
(1,1) & 
\left(\frac{t+4}{2}, \frac{t+4}{2}\right) &
\left(\frac{5-t}{2}, \frac{5-t}{2}\right) \\
\left(\frac{5-t}{2}, \frac{5-t}{2}\right) &
\left(\frac{t+4}{2}, \frac{t+4}{2}\right) &
\left(-t^2-t+3, -t^2-t+3\right) &
\left(t (t+4),t^2-6 t+5\right) \\
\left(\frac{t+4}{2}, \frac{t+4}{2}\right) &
\left(\frac{5-t}{2}, \frac{5-t}{2}\right)  &
\left(t^2-6 t+5,t (t+4)\right) &
\left(-t^2+3 t+1, -t^2+3 t+1\right)
\end{pmatrix}.
\end{equation}
In particular, let us determine a specific form of the game for $\theta_1=\pi/3$, or equivalently, for $t=3/4$. We get
\begin{equation}
\renewcommand\arraystretch{1.5}
C=\bordermatrix{ & I & iX & U_{1} & U_{2} \cr 
I & (3,3) & (0,5) & \left(\frac{17}{8},\frac{17}{8}\right) &  \left(\frac{19}{8},\frac{19}{8}\right) \cr
iX &  (5,0) & (1,1) &  \left(\frac{19}{8},\frac{19}{8}\right) & \left(\frac{17}{8},\frac{17}{8}\right)\cr
U_{1} & \left(\frac{17}{8},\frac{17}{8}\right) &  \left(\frac{19}{8},\frac{19}{8}\right) & \left(\frac{27}{16},\frac{27}{16}\right) &  \left(\frac{57}{16},\frac{17}{16}\right) \cr
U_{2} & \left(\frac{19}{8},\frac{19}{8}\right) &  \left(\frac{17}{8},\frac{17}{8}\right) & \left(\frac{17}{16},\frac{57}{16}\right) &  \left(\frac{43}{16},\frac{43}{16}\right)    
}.
\end{equation}
The above extension of the Prisoner's Dilemma game significantly enriches the game and changes its Nash equilibria. While the original PD has only a single equilibrium with a payoff 1 for both players, corresponding to the strategy profile $(iX,iX)$, this extended version offers multiple equilibrium points. Two equilibria with a payoff $\frac{19}{8}$ corresponding to the profiles $(iX,U_1)$ and $(U_1,iX)$ and one mixed strategy equilibrium with a payoff $\frac{23}{12}$, corresponding to the strategy $(0,\frac{1}{3}iX,\frac{2}{3}U_1,0)$ for both players.

\subsection{Extension of the D class}
Here again $\theta_1, \theta_2 \in (0,\pi)$ with the condition that $\theta_2=\pi - \theta_1$. 
The parameters $\alpha_1$ and $\beta_1$ are mutually dependent, namely $\beta_1 = \alpha_1 + n \pi$, where $n \in \mathbb{Z}$. The set of all such solutions is given by Eq. (\ref{solDD}). There are two extension matrices in this class
\begingroup 
\setlength\arraycolsep{4.5pt}
\renewcommand\arraystretch{1.5}
\begin{equation}\label{klasaD}
\medmuskip = 0.2mu
D_1 = \begin{pmatrix}
\Gamma & t\Gamma+t'\Gamma^2 \\ 
t\Gamma+t'\Gamma^1 & t^2\Gamma + tt'(\Gamma^1+\Gamma^2) + t'^2\Gamma^3
\end{pmatrix}, \quad 
D_2 = \begin{pmatrix}
\Gamma & t\Gamma^3+t'\Gamma^1 \\ 
t\Gamma^3+t'\Gamma^2 & t^2\Gamma + tt'(\Gamma^1+\Gamma^2) + t'^2\Gamma^3
\end{pmatrix}
\end{equation}
\endgroup
where $t=\cos^{2}\frac{\theta_1}{2}$, $t'= 1-t =\sin^{2}\frac{\theta_1}{2}$, and $\alpha_i, \beta_j \in \{0, \pi\}$ for $D_1$ or $\alpha_i, \beta_j \in \{\frac{\pi}{2}, \frac{3\pi}{2}\}$  for $D_2$. The extensions $D$ converge to the extensions $A$ as $\theta_1$ approaches $0$ or $\pi$: 
 \begin{equation}
\lim_{\theta_{1} \to 0} D_{1} = \left. A_1\right|_{\alpha_{1} =0}, \quad \lim_{\theta_{1}\to 0}D_{2} = \left. A_1\right|_{\alpha_{1} =\frac{\pi}{2}}
 \end{equation} and
 \begin{equation}
\lim_{\theta_{1} \to \pi} D_{1} = \left. A_2\right|_{\alpha_{2} =0}, \quad \lim_{\theta_{1}\to \pi}D_{2} = \left. A_2\right|_{\alpha_{2} =\frac{\pi}{2}}.
 \end{equation} 
\subsection{Extension of the E class}
For this type $\theta_1, \theta_2 \in (0,\pi)$, where again $\theta_2=\pi - \theta_1$. 
The parameters $\alpha_1$ and $\beta_1$ obey: $\beta_1 = \alpha_1 + (n + \frac{1}{2}) \pi$, where $n \in \mathbb{Z}$. The set of all such parameters is given by Eq. (\ref{solEE}). The extended game matrices are 
\begingroup 
\setlength\arraycolsep{4.5pt}
\renewcommand\arraystretch{1.5}
\begin{equation}
\medmuskip = 0.2mu
E_1 = \begin{pmatrix}\label{klasaE}
\Gamma & t\Gamma+t'\Gamma^1 \\ 
t\Gamma+t'\Gamma^2 & t^2\Gamma + tt'(\Gamma^1+\Gamma^2) + t'^2\Gamma^3
\end{pmatrix}, \quad
E_2 = \begin{pmatrix}
\Gamma & t\Gamma^3+t'\Gamma^2 \\ 
t\Gamma^3+t'\Gamma^1 & t^2\Gamma + tt'(\Gamma^1+\Gamma^2) + t'^2\Gamma^3
\end{pmatrix},
\end{equation}
\endgroup
where $t=\cos^{2}\frac{\theta_1}{2}$, $t'= 1-t =\sin^{2}\frac{\theta_1}{2}$ and moreover $\alpha_1, \beta_2 \in \{0, \pi\}$, $\alpha_2, \beta_1 \in \{\frac{\pi}{2}, \frac{3\pi}{2}\}$ for $E_1$, whereas $\alpha_1, \beta_2 \in \{\frac{\pi}{2}, \frac{3\pi}{2}\}$, $\alpha_2, \beta_1 \in \{0, \pi\}$ for $E_2$. Here again the extensions $E$ converge to the extensions $A$ as $\theta_1$ approaches $0$ or $\pi$:
 \begin{equation}
\lim_{\theta_{1} \to 0} E_{1} = \left. A_1\right|_{\alpha_{1} =0}, \quad \lim_{\theta_{1}\to 0}E_{2} = \left. A_1\right|_{\alpha_{1} =\frac{\pi}{2}}
 \end{equation} and
 \begin{equation}
\lim_{\theta_{1} \to \pi} E_{1} = \left. A_2\right|_{\alpha_{2} =\frac{\pi}{2}}, \quad \lim_{\theta_{1}\to \pi}E_{2} = \left. A_2\right|_{\alpha_{2} = 0}.
 \end{equation} 
The convergence of extensions D and E to extension A at $\theta_{1} \to 0$ and $\theta_{1} \to \pi$, is depicted in Figure 1 while Table 1 presents the ranges of all parameters of allowable extensions of the classic game (\ref{bimatrixiso}).
    \begin{figure}[t]
    \centering 
\includegraphics[width=4in]{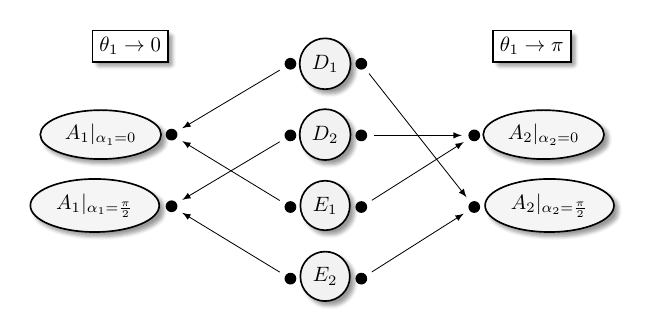}
\caption{The convergence of extensions D and E to extension A at $\theta_{1} \to 0$ and $\theta_{1} \to \pi$. \label{fig1}}
\end{figure}

\NewDocumentCommand{\INTERVALINNARDS}{ m m }{
    #1 {,} #2
}
\NewDocumentCommand{\interval}{ s m >{\SplitArgument{1}{,}}m m o }{
    \IfBooleanTF{#1}{
        \left#2 \INTERVALINNARDS #3 \right#4
    }{
        \IfValueTF{#5}{
            #5{#2} \INTERVALINNARDS #3 #5{#4}
        }{
            #2 \INTERVALINNARDS #3 #4
        }
    }
}
\begingroup 
\begin{table}
\renewcommand*{\arraystretch}{1.7}%
\begin{tabular}{ | M{0.7cm} | M{0.9cm}| M{0.9cm} | M{1.8cm} | M{1.8cm} |M{1.8cm} |M{1.8cm} |M{1 cm} |M{1.35cm}|} 

  \hline
   Class & $\theta_1$ & $\theta_2$ & $\alpha_1$ & $\beta_1$ &$\alpha_2$ &  $\beta_2$ & Eq. for solutions & number of discrete solutions \\ 
  \hline
  \hline
  A & {$\!\begin{aligned}
 0 \\
 \pi
\end{aligned} $}& {$\!\begin{aligned}
 \pi \\
 0
\end{aligned}$} & {$\!\begin{aligned}
 \interval[{0,2\pi}) \\
 \interval[{0,2\pi})
\end{aligned}$} & {$\!\begin{aligned}
 \interval[{0,2\pi}) \\
 \interval[{0,2\pi})
\end{aligned}$}& {$\!\begin{aligned}
 \interval[{0,2\pi}) \\
 n\pi-\beta_1
\end{aligned}$}& {$\!\begin{aligned}
 n\pi-\alpha_1 \\
\interval[{0,2\pi})
\end{aligned}$} &  {$\!\begin{aligned}
 (\ref{t10}) \\
 (\ref{t1pi})
\end{aligned}$} & {$\!\begin{aligned}
 1 \\
 1
\end{aligned}$} \\
  \hline
  B & $\frac{\pi}{2}$  & $\frac{\pi}{2}$ & $\frac{\pi}{4},\frac{3\pi}{4},\frac{5\pi}{4}, \frac{7\pi}{4}$& $\frac{\pi}{4},\frac{3\pi}{4},\frac{5\pi}{4}, \frac{7\pi}{4}$& $\beta_1 + n \pi$  &  $\alpha_1 + l \pi $ & (\ref{solBB}) & 64 \\ 
  \hline
  C & $(0,\pi)$  & $\pi-\theta_1$ & $\frac{\pi}{4},\frac{3\pi}{4},\frac{5\pi}{4}, \frac{7\pi}{4}$ & $\frac{\pi}{4},\frac{3\pi}{4},\frac{5\pi}{4}, \frac{7\pi}{4}$  & $\beta_1 + (n +\frac{1}{2})\pi$ & $\alpha_1 + (l +\frac{1}{2})\pi$ & (\ref{solCC}) & 64\\ 
  \hline
  D & $(0,\pi)$  & $\pi-\theta_1$   &0, $\frac{\pi}{2},\pi,\frac{3\pi}{2}$ & $\alpha_1 + n \pi $ & $\beta_1 + l \pi $&  $\alpha_1 + m \pi $&  (\ref{solDD}) & 32 \\ 
  \hline
  E & $(0,\pi)$  & $\pi-\theta_1$   & 0, $\frac{\pi}{2},\pi,\frac{3\pi}{2}$ & $\alpha_1 + (n +\frac{1}{2})\pi$ & $\beta_1 + l \pi $  &  $\alpha_1 + m \pi $ & (\ref{solEE}) & 32 \\ 
  \hline
\end{tabular}
\caption{Parameters of permissible extensions A - E of the game defined by the bimatrix (\ref{bimatrixiso}), where $n, l, m \in \mathbb{Z}$ and $\alpha_i, \beta_j \in [0,2\pi)$.} 
\end{table}
\endgroup

\section{Conclusions}
The aim of the work was to determine all possible pairs of operators that, together with classical strategies, create a $4\times 4$ quantum game invariant under isomorphic transformations of the classical game. Our research showed that the problem of finding two unitary strategies was much more complex that the problem regarding a single unitary strategy considered in \cite{frackiewicz_permissible_2024}. We proved a theorem giving a practical criterion for the invariance of the quantum extension with respect to isomorphic transformations of the classical game. As a result, we have determined five classes of games and all unitary operations corresponding to these classes. Each game class returns the same game theory problem for a given input classical game and its isomorphic counterparts. We have identified the interdependencies between different classes of extensions, including situations where one class evolves into another.

The exploration of quantum game theory enriches our understanding of strategic behavior in complex systems.  By providing a framework for analyzing and predicting the outcomes of interactions among rational agents operating under quantum rules, this field paves the way for the development of new strategies for cooperation, competition, and conflict resolution in a world of quantum computers. Moreover, the application of quantum game theory to decision-making in strategic interactions reveals novel insights into the optimization of quantum algorithms, potentially revolutionizing computational methods and technologies. In summary, the study of quantum game theory, particularly through the EWL approach and the analysis of mixed strategies involving multiple simple quantum strategies, represents a crucial step forward in the integration of quantum computation with strategic decision-making. Its exploration not only broadens the theoretical horizons of game theory but also offers tangible benefits for the advancement of quantum information processing, with implications for economics, cybersecurity, and beyond.

\section*{Acknowledgements}
Computations were carried out using the computers of Centre of Informatics Tricity Academic Supercomputer \& Network.
\bibliographystyle{qipstyle}
\bibliography{references-2}

\end{document}